\documentclass[english]{article}
\usepackage[T1]{fontenc}
\usepackage[latin9]{inputenc}
\usepackage{verbatim}
\usepackage{amsthm}
\usepackage{amsmath}
\usepackage{amssymb}

\makeatletter
\theoremstyle{plain}
\newtheorem{thm}{\protect\theoremname}
  \theoremstyle{definition}
  \newtheorem{defn}{\protect\definitionname}
  \theoremstyle{plain}
  \newtheorem*{thm*}{\protect\theoremname}
  \theoremstyle{plain}
  \newtheorem{prop}{\protect\propositionname}

\usepackage{complexity}

\newenvironment{theorem_again}[1]{\noindent{\bf Theorem #1.~}}{}

\makeatother

\usepackage{babel}
  \providecommand{\definitionname}{Definition}
  \providecommand{\propositionname}{Proposition}
  \providecommand{\theoremname}{Theorem}
\providecommand{\theoremname}{Theorem}

\begin{document}

\title{Computational Complexity of Approximate Nash Equilibrium in Large
Games}

\author{Aviad Rubinstein\thanks{UC Berkeley.
I am grateful to Christos Papadimitriou for inspiring discussions, comments, and advice.
I would also like to thank Constantinos Daskalakis and Paul Goldberg for pointing out important missing references in an earlier version.}}
\maketitle
\begin{abstract}
We prove that finding an $\epsilon$-Nash equilibrium in a succinctly
representable game with many players is \PPAD-hard for constant $\epsilon$.
Our proof uses {\em succinct games}, i.e. games whose (exponentially-large)
payoff function is represented by a smaller circuit. Our techniques
build on a recent query complexity lower bound by Babichenko \cite{Bab13_query_complexity}.
\end{abstract}

\section{Introduction}

Nash equilibrium is the central concept in Game Theory. Much of its
importance and attractiveness comes from its {\em universality}:
by Nash's Theorem \cite{Nash}, every finite game has at least one.
The result that finding a Nash equilibrium is \PPAD-complete, and
therefore intractable \cite{NASH-is-PPAD-hard_DGP09,2-player_nash_CDT09}
casts this universality in doubt, since it suggests that there are
games whose Nash equilibria, though existent, are in any practical
sense inaccessible.

Can approximation repair this problem? Chen et al. \cite{2-player_nash_CDT09}
proved that it is also hard to find an $\epsilon$-Nash equilibrium
for any $\epsilon$ that is polynomially small - even for two-player
games. The only remaining hope is a PTAS, i.e. an approximation scheme
constant $\epsilon>0$. {\em Whether there is a PTAS for the Nash
equilibrium problem is the most important remaining open question
in equilibrium computation.} When we say ``$\epsilon$-Nash equilibrium,''
we mean the additive sort; for multiplicative $\epsilon$-Nash equilibria,
Daskalakis \cite{Das13_multiplicative_hardness} shows that the problem
is \PPAD-hard even for two-player games; notice that such a result
is unlikely for additive approximation and a constant number of players,
since a quasi-polynomial time approximation algorithm exists \cite{LMM03_quasi_poly}.

\subsubsection*{Our results}

In this paper we make a modest step towards the inapproximability,
in the standard additive sense, of Nash equilibrium:
\begin{thm}
\label{thm:pwsn}There exist constants $\epsilon,k>0$, such that
given a game with $m$ players and $k$ actions, and a poly-size circuit
that computes the vector of payoffs for each vector of pure strategies,
it is \PPAD-hard to compute an $\epsilon$-Nash equilibrium.
\end{thm}
Even though it is the first result establishing inapproximability
of an additive notion of Nash equilibrium, it has two deficiencies:
(a) it is about {\em$\epsilon$-Nash-equilibrium} (sometimes%
\footnote{See \cite{Das13_multiplicative_hardness} for a short discussion on
terminology.%
} also called {\em $\epsilon$-well-supported Nash equilibrium})
\cite{NASH-is-PPAD-hard_DGP09}, requiring that all actions in the
support be approximately best responses (instead of the mixture being
best response in $\epsilon$-approximate Nash equilibrium); and (b)
it holds for a somewhat awkward class of multiplayer games we call
{\em succinct games}%
\footnote{There is some disagreement in the literature about the terminology:
a very similar definition appeared in \cite{SV12_succint_games} as
{\em circuit games}; \cite{FKS95_poly-definable-games} discuss
{\em polynomially definable games}; while \cite{FIKU08_succinct-zero-sum}
did use the term succinct games.%
}, in which the utility of each player for an action profile is calculated
by a circuit. For example, on such games computing the exact best
response is \PP-hard.

En route to proving Theorem \ref{thm:pwsn}, we also prove a similar
statement about the computational complexity of finding an approximate
fixed-point for a continuous (in fact, Lipschitz) function, which
may be of separate interest.
\begin{thm}
\label{thm:pafp}There exist constants $\epsilon,M>0$, such that
given a $M$-Lipschitz, poly-time computable function $f\colon\left[0,1\right]^{n}\rightarrow\left[0,1\right]^{n}$,
it is \PPAD-hard to find an $\epsilon$-approximate fixed point of
$f$.
\end{thm}

\paragraph{Pure equilibria and \PPAD-hardness }

It is interesting to note that all our hard game instances {\em have
pure equilibria}. To the best of our knowledge, this is the first
setting where \PPAD-hardness is proved for games which have pure
equilibria. Naturally, in any game with a poly-size payoff matrix,
it is easy to find any pure equilibrium.

\subsection{Related works}

Schoenebeck and Vadhan \cite{SV12_succint_games} studied comprehensively
the computational complexity of Nash Equilibria in succinct games
and other concise representations. They characterized the computational
complexity, for different levels of approximation, of the following
questions: is a given strategy profile a Nash equilibium? does there
exist a pure equilibrium? does there exist a Nash equilibrium in which
the payoff to a given player meets a given guarantee? They also prove
that the problem of finding an $\epsilon$-Nash equilibrium (for constant
$\epsilon$) in circuit games where each player has two strategies
is \BPP-hard (for a promise-problem version of BPP), and belongs
to $\P^{\MA}$. Finding a tighter complexity classification is stated
as an open question. 

We note that different variants of succinct games were also studied
by \cite{PY94_bounded-rationality,FKS95_poly-definable-games,FIKU08_succinct-zero-sum}.

\subsubsection*{Related works on query complexity}

There are several interesting results on the query complexity of approximate
Nash equilibria, where the algorithm is assumed to have black-box
access to the exponential-size payoff function. In other words, in
this setting the payoff function is allowed to be arbitrarily complex.

A recent paper by Hart and Nisan \cite{HN13_query_complexity_correlated}
proves that any deterministic algorithm needs to query at least an
exponential number of queries to compute any $\epsilon$-Nash equilibrium
- and even for any $\epsilon$-correlated equilibrium. For $\epsilon$-correlated
equilibrium, on the other hand, Hart and Nisan show a randomized algorithm
that uses a number of queries polynomial in $n$ and $\epsilon^{-1}$.

Of particular interest to us, is a very recent paper by Babichenko
\cite{Bab13_query_complexity}, that extends the hardness result of
Hart and Nisan to show that any randomized algorithm still requires
an exponential number of queries to find an $\epsilon$-Nash equilibrium.
Our proof is inspired by Babichenko's work and builds on his techniques.

Finally, yet a newer paper by Goldberg and Roth \cite{GR14_query_complexity_concise-WSNE}
characterizes the query complexity of approximate coarse correlated
equilibrium in games with many players. More important for our purpose
is their polynomial upper bound on the query complexity of $\epsilon$-Nash
equilibria for any family of games that have {\em any} concise representation.
This result is to be contrasted with (a) Babichenko's query complexity
lower bound, which uses a larger family of games, and (b) our result
which applies exactly to this setting and gives a lower bound on the
{\em computational complexity}. 

A much older yet very interesting and closely related result is that
of Hirsch, Papadimitriou, and Vavasis \cite{HPV89}. Hirsch et al.
show that any deterministic algorithm for computing a Brouwer fixed
point in the oracle model must make an exponential -in the dimension
$n$ and the approximation $\epsilon$- number of queries for values
of the function. The techniques in \cite{HPV89} have proven particularly
useful both in Babichenko's work, and in ours.

\subsubsection*{Related works on succinct representations of other objects}

Although our notion of succinct representation is somewhat non-standard
in game theory, similar succinct problems have been considered before.
Galperin and Wigderson \cite{GW84_succint_graphs} and Papadimitriou
and Yannakakis \cite{PY86_succint_graphs} studied the computation
of graph properties for exponential-size graphs given access to a
circuit that locally computes the adjacency matrix.

Much more recently, Dobzinski and Vondrak studied optimization of
succinctly-represented objects for submodular optimization \cite{DV12_computational_submodular}
and for combinatorial auctions \cite{DV12_computational_auctions}.
As in our case, in both of those settings similar hardness results
were previously known in the value oracle model.

\subsection{Techniques\label{sec:Techniques}}

Our techniques in this paper are significantly different from previous
works on \PPAD-hardness in games such as Daskalakis et al. \cite{NASH-is-PPAD-hard_DGP09}
and Chen et al. \cite{2-player_nash_CDT09}. In particular two ingredients
of our reduction are fundamentally different from the so-called ``DGP
framework'': the construction of the Brouwer function, and the reduction
from Brouwer to Nash. We note that in both cases our techniques are
much closer to the recent work of Babichenko \cite{Bab13_query_complexity}.

\paragraph{From Brouwer to Nash}

Our reduction from Brouwer to Nash follows an argument presented in
Eran Shmaya's blog \cite{Shm12_blog}. Each player has actions in
$\left[0,1\right]$, and the players are divided into two groups:
the first group tries to imitate the second group, while the second
group tries to imitate the Brouwer function applied to the actions
chosen by the first group.

In \cite{NASH-is-PPAD-hard_DGP09,2-player_nash_CDT09}, per contra,
the probability assigned to each action corresponds, roughly, to a
variable in $\left[0,1\right]$. While this construction is more stringent
on the number of actions, it requires a relatively complex averaging
gadget. More importantly, this averaging gadget as implemented by
Chen et al. seems to require a polynomial blow-up of the error. Thus,
it is not clear how to achieve a constant hardness of approximation
in this way.

\paragraph{The Brouwer function}

In order to construct the hard Brouwer function we use a construction
due to Hirsch, Papadimitriou, and Vavasis \cite{HPV89} that embeds
paths as mappings over $\left[0,1\right]^{n}$, in a delicate way
which we describe later.

Chen et al. \cite{2-player_nash_CDT09}, on the other hand, divide
the unit hypercube into subcubes of constant edge-length, and specify
a color for each subcube. The $i$-th color corresponds to $\xi_{i}$,
i.e. the $i$-th unit vector, whereas a special red color corresponds
to $-\sum\xi_{i}$. Any vector of players' mixed strategies corresponds
to a distribution over neighboring subcubes, and we get a fixed point
whenever the expectation of the corresponding vectors is $0^{n}$.

Clearly, any distribution corresponding to an exact fixed point must
have support over a panchromatic neighborhood, i.e. a neighborhood
with subcubes corresponding to all $n+1$ colors. Finally, it is shown
that it is \PPAD-hard to find such panchromatic neighborhood. In
fact, a panchromatic neighborhood is still necessary for any $\Theta\left(1/n\right)$-approximate%
\footnote{Approximation here is in the sense of $L^{\infty}$. In fact, it seems
to give a constant hardness of approximation in $L^{1}$, but unfortunately
it is not clear how to use that for a reduction to Nash equilibria. %
} fixed point. However, a $\left(1/n\right)$-approximate fixed point
can be achieved from a distribution over only $n$ colors: for example,
taking each of the colors $i\in\left[n\right]$ with probability $1/n$
results in the expected vector $\left(1/n,\dots,1/n\right)$.

\section{Preliminaries}

Throughout this paper we use the max-norm as the default measure of
distance. In particular, when we say that $f$ is $M$-Lipschitz we
mean that for every $\mathbf{x}$ and \textbf{$\mathbf{y}$ }in the
domain of $f$, $\left\Vert f\left(\mathbf{x}\right)-f\left(\mathbf{y}\right)\right\Vert _{\infty}\leq M\left\Vert \mathbf{x}-\mathbf{y}\right\Vert _{\infty}$.

\subsubsection*{The {\sc EndOfTheLine} problem}

Our reduction starts form the {\sc EndOfTheLine} problem. This problem
was implicit in \cite{PPAD_Pap94}, and explicitly by defined Daskalakis
et al. \cite{NASH-is-PPAD-hard_DGP09}.
\begin{defn}
{\sc EndOfTheLine}: (\cite{NASH-is-PPAD-hard_DGP09}) Given two
circuits $S$ and $P$, with $n$ input bits each, such that $P\left(0^{n}\right)=0^{n}\neq S\left(0^{n}\right)$,
find an input $x\in\left\{ 0,1\right\} ^{n}$ such that $P\left(S\left(x\right)\right)\neq x$
or $S\left(P\left(x\right)\right)\neq x\neq0^{n}$.\end{defn}
\begin{thm*}
(Essentially \cite{PPAD_Pap94}) {\sc EndOfTheLine} is \PPAD-complete
(for poly-size $S$ and $P$).
\end{thm*}
\begin{comment}

\subsection{The Hirsch-Papadimitriou-Vavasis mapping}

Our reduction uses a modification of the mapping constructed by Hirsch
et al. \cite{HPV89}. We recall some of the high level ideas of this
construction.
\end{comment}

\subsubsection*{Succinct normal form games}

We consider games with $n$ players. Player $i$ chooses one of $k$
actions $a_{i}\in A_{i}$. The utility of player $i$ for each vector
of actions is given by $u_{i}\colon\times_{j}A_{j}\rightarrow\left[0,1\right]$.

Explicitly describing the $u_{i}$'s requires space exponential in
$n$. In this paper we restrict our attention to games that can be
described {\em succinctly}; namely, there is a poly-size circuit
that computes each of the $u_{i}$'s given a vector of actions $\mathbf{a}\in\times_{j}A_{j}$.

\subsubsection*{$\epsilon$-Nash equilibrium vs $\epsilon$-approximate Nash equilibrium}

A mixed strategy of player $i$ is a distribution $x_{i}\in\Delta A_{i}$.
We say that a vector of mixed strategies \textbf{$\mathbf{x}\in\times_{j}\Delta A_{j}$}
is a {\em Nash equilibrium} if every strategy $a_{i}$ in the support
of $x_{i}$ is a best response to the actions of the mixed strategies
of the rest of the players, $x_{-i}$. Formally, 
\[
\mathbb{E}_{a_{-i}\sim x_{-i}}\left[u_{i}\left(a_{i},a_{-i}\right)\right]=\max_{a'\in A_{i}}\mathbb{E}_{a_{-i}\sim x_{-i}}\left[u_{i}\left(a',a_{-i}\right)\right]\,.
\]
Equivalently, $\mathbf{x}$ is a Nash equilibrium if each mixed strategy
$x_{i}$ is a best mixed response to $x_{-i}$:
\[
\mathbb{E}_{\mathbf{a}\sim\mathbf{x}}\left[u_{i}\left(\mathbf{a}\right)\right]=\max_{x_{i}'\in\Delta A_{i}}\mathbb{E}_{\mathbf{a}\sim\mathbf{x'}}\left[u_{i}\left(\mathbf{a}\right)\right]\,.
\]

Each of those equivalent definitions can be generalized to include
approximation in a different way. (Of course, there are also other
interesting generalizations of Nash equilibria to approximate settings.)
We say that $\mathbf{x}$ is an {\em $\epsilon$-approximate Nash
equilibrium} ({\em $\epsilon$-ANE}) if each $x_{i}$ is an $\epsilon$-best
mixed response to $x_{-i}$:
\[
\mathbb{E}_{\mathbf{a}\sim\mathbf{x}}\left[u_{i}\left(\mathbf{a}\right)\right]\geq\max_{x_{i}'\in\Delta A_{i}}\mathbb{E}_{\mathbf{a}\sim\mathbf{x'}}\left[u_{i}\left(\mathbf{a}\right)\right]-\epsilon\,.
\]

On the other hand, we generalize the first definition of Nash equilibrium
by saying that $\mathbf{x}$ is a {\em $\epsilon$-Nash equilibrium}
({\em $\epsilon$-NE}; sometimes also $\epsilon$-well-supported
Nash equilibrium) if each $a_{i}$ in the support of $x_{i}$ is an
$\epsilon$-best response to $x_{-i}$:
\[
\mathbb{E}_{a_{-i}\sim x_{-i}}\left[u_{i}\left(a_{i},a_{-i}\right)\right]\geq\max_{a'\in A_{i}}\mathbb{E}_{a_{-i}\sim x_{-i}}\left[u_{i}\left(a',a_{-i}\right)\right]-\epsilon\,.
\]

It is easy to see every $\epsilon$-NE is also an $\epsilon$-ANE,
but the converse is false. Given an $\epsilon$-ANE it is possible
to find a $\Theta\left(\sqrt{\epsilon}n\right)$-NE (see e.g. \cite{NASH-is-PPAD-hard_DGP09});
however computational hardness of $n^{-c}$-approximate Nash equilibrium
is a corollary of \cite{2-player_nash_CDT09}.

\section{Main result}
\begin{defn}
Given a game over $m$ players with $k$ actions, and a poly-size
circuit that computes the vector of payoffs for each vector of pure
strategies, {\sc SuccinctNash}$\left(n,k,\epsilon\right)$ is the
problem of computing an $\epsilon$-NE for this game.
\end{defn}
\begin{theorem_again}{\ref{thm:pwsn}} {\sc SuccinctNash}$\left(n,10^{4},10^{-8}\right)$
is \PPAD-hard.

\end{theorem_again}

\subsubsection*{Proof overview}

We begin our proof with the {\sc EndOfTheLine} problem on $\left\{ 0,1\right\} ^{n}$
\cite{NASH-is-PPAD-hard_DGP09}. In the first step, we embed the {\sc EndOfTheLine}
as a collection $H$ of vertex-disjoint paths on the $\left(2n+1\right)$-dimensional
hypercube graph. Given $H$, our second step is to construct a continuous
mapping $f\colon\left[0,1\right]^{2n+2}\rightarrow\left[0,1\right]^{2n+2}$
whose fixed points correspond to ends of paths in $H$. This step
is done using a technique introduced by Hirsch et al. \cite{HPV89}.
Our third and final step is to reduce the problem of finding approximate
fixed points of $f$ to the problem of finding approximate NE in a
$4n+4$-players game via a reduction which appeared in Shmaya's blog
\cite{Shm12_blog}.

\subsection{Embedding the {\sc EndOfTheLine} graph as paths in $\left\{ 0,1\right\} ^{2n+1}$}

Our first step in the reduction is to embed an {\sc EndOfTheLine}
graph as vertex-disjoint paths on the $\left(2n+1\right)$-dimensional
hypercube graph. We first recall that the input to the {\sc EndOfTheLine}
problem is given as two circuits $S$ and $P$, which define a directed
graph over $G$ over $\left\{ 0,1\right\} ^{n}$. Given $S$ and $P$,
we construct a collection $H$ of vertex-disjoint paths over the $\left(2n+1\right)$-dimensional
hypercube graph, such that each starting or end point of a path in
$H$ corresponds to a unique starting or end point of a line in $G$. 

In order to construct our embedding we divide the $2n+1$ coordinates
as follows: the first $n$ coordinates store current vertex $\mathbf{u}$,
the next $n$ coordinates for next vertex in the line, $\mathbf{v}$,
and finally, the last coordinate $b$ stores a compute-next vs copy
bit. When $b=0$, the path proceeds to update $\mathbf{v}\leftarrow S\left(\mathbf{u}\right)$,
bit-by-bit. When this update is complete, the value of $b$ is changed
to $1$. Whenever $b=1$, the path proceeds by copying $\mathbf{u}\leftarrow\mathbf{v}$
bit-by-bit, and then changes that value of $b$ again. Finally, when
$\mathbf{u}=\mathbf{v}=S\left(\mathbf{u}\right)$ and $b=0$, the
path reaches an end point.

Notice that the paths in $H$ do not intersect. Furthermore, given
a vector in $\mathbf{p}\in\left\{ 0,1\right\} ^{2n+1}$, we can output
in polynomial time whether $\mathbf{p}$ belongs to a path in $H$,
and if so which are the previous and consecutive vectors in the path.
It is therefore \PPAD-hard to find a starting or end point of any
path in $H$ other than $0^{n}$.

\subsection{Continuous mapping on $\left[0,1\right]^{2n+2}$}

Our second step, which constructs a continuous mapping given $H$,
is probably the most technically involved. Fortunately, almost all
of the technical work we need was already done by Hirsch et al. \cite{HPV89}. 
\begin{defn}
Given an $M$-Lipschitz, poly-time computable function $f\colon\left[0,1\right]^{n}\rightarrow\left[0,1\right]^{n}$,
{\sc SuccinctBrouwer}$\left(n,M,\epsilon\right)$ is the problem
of computing an $\epsilon$-approximate fixed point of $f$.
\end{defn}
\begin{theorem_again}{\ref{thm:pafp}} {\sc SuccinctBrouwer}$\left(n,80,1/88\right)$
is \PPAD-hard.

\end{theorem_again}
\begin{proof}
We begin with a quick overview of the mapping constructed by Hirsch
et al. \cite{HPV89}. We will then show how to adapt their construction
to fit our reduction. In the following, we denote $g\left(\mathbf{x}\right)=f\left(\mathbf{x}\right)-\mathbf{x}$.

\subsubsection*{The HPV mapping}

Given a path in $\left\{ 0,1\right\} ^{2n+1}$, Hirsch et al. \cite{HPV89}
construct a mapping $f\colon\left[0,1\right]^{2n+2}\rightarrow\left[0,1\right]^{2n+2}$
that satisfies%
\footnote{See also \cite{Bab13_query_complexity} for the choice of constants.
Please note the change in variable names: $h,M,2^{-p}$ in \cite{HPV89}
correspond to $\delta,\lambda,\epsilon$, respectively, in \cite{Bab13_query_complexity}.%
}:
\begin{enumerate}
\item $g$ is $79$-Lipschitz (thus, $f$ is $80$-Lipschitz)
\item $\left\Vert g\left(\mathbf{x}\right)\right\Vert _{\infty}\geq1/88$
for every $\mathbf{x}$ that does not correspond to the endpoint of
the path
\item The value of $g$ at each point $\mathbf{x}$ depends only on whether
the path passes through the subcube corresponding to $\mathbf{x}$,
and in which direction.
\end{enumerate}
However, for our purposes, it does not suffice to embed just a single
path. In particular, we have one path whose starting point we know
(and must not correspond to a fixed point), and many other paths (and
cycles), whose end points -and starting points- correspond to additional
fixed points.

Luckily, the same construction of \cite{HPV89} continues to work
in our case, with only minor changes. In order to explain those modifications,
we first briefly recall some of the details in the original construction.

We divide the $\left[0,1\right]^{2n+2}$ hypercube into smaller subcubes
of edge-size $h$. In \cite{HPV89}, a single path over $\left\{ 0,1\right\} ^{2n+1}$
corresponds to a $\left(2n+1\right)$-dimensional sequence of subcubes,
called the {\em tube}, that all lie in a special designated {\em slice}
of $\left[0,1\right]^{2n+2}$. 

The {\em home subcube}, the subcube that corresponds to the beginning
of the path, is special: the flow from all subcubes that do not belong
to the path leads to this subcube. 

For the purpose of adapting this construction, the most important
property is that on (and near) the outer facets of the tube, i.e.
the facets that do not face other subcubes in the path, $g\left(\mathbf{x}\right)=\delta\mathbf{\xi}_{2n+2}$,
where $\mathbf{\xi}_{2n+2}$ is the $\left(2n+2\right)$-unit vector,
and $\delta$ is some small parameter (constant in this paper). The
same also holds for all points in the slice that do not belong to
the tube. Intuitively, this means that {\em all subcubes -whether
they belong to the tube or not- look the same from the outside} (except
for the two facets that continue the path).

\subsubsection*{Embedding multiple paths}

Given a collection $H$ of non-intersecting paths in $\left\{ 0,1\right\} ^{2n+1}$,
we construct a mapping $f\colon\left[0,1\right]^{2n+2}\rightarrow\left[0,1\right]^{2n+2}$
in a similar fashion. Essentially, we construct many tubes (some of
which may form close cycles) in the same slice. 

If any path in $H$ passes through some $\mathbf{p}\in\left\{ 0,1\right\} ^{2n+1}$,
then $\mathbf{p}$ corresponds to a subcube on which $f$ is defined
exactly the same as in \cite{HPV89}. Likewise, every end point of
path in $H$ corresponds to a subcube on which $f$ is defined exactly
the same as the unique end point in \cite{HPV89}.

Our construction differs in the starting points of the paths in $H$.
In \cite{HPV89}, the starting point corresponds to the home subcube,
which is \emph{universally unique}: the flow from any subcube outside
the path leads towards that point. Indeed we cannot imitate this behaviour
for the starting point of every path in $H$. However, this does not
complicate our reduction, because all other starting points are also
(\PPAD-) hard to find. In particular we can construct $f$ in a similar
fashion to the end points of the paths, thereby creating additional
fixed points.

For each starting point, consider the corresponding subcube of $\left[0,1\right]^{2n+2}$.
The values of $f$ on the facet $F_{1}$ in the direction of the path,
are already determined by the next subcube in the path. Now, we let
$g\left(\mathbf{x}\right)=\delta\mathbf{\xi}_{2n+2}$ uniformly for
every point $\mathbf{x}$ on the opposite facet, $F_{0}$. For any
$\mathbf{x}'$ between the facets we interpolate by taking the weighted
average of the values $f\left(\mathbf{x}_{0}^{'}\right)$ and $f\left(\mathbf{x}_{1}^{'}\right)$
on the projections of $\mathbf{x}'$ on $F_{0}$ and $F_{1}$, respectively.
(This corresponds to \cite{HPV89}'s ``Cartesian interpolation'').
Notice that this subcube also satisfies $g\left(\mathbf{x}\right)=\delta\mathbf{\xi}_{2n+2}$
for all points on facets that do not face the rest of the path.

Outside the subcubes corresponding to the additional paths, all the
properties of the mapping are trivially preserved. On the interface
between any subcube in an additional path, and any other subcube which
is not consecutive in the same path, all the properties are again
preserved since adding the paths does not change the value of $f$
near those facets. Finally within each path, $f$ is constructed exactly
the same way as the single path in \cite{HPV89}, and therefore it
is easy to see that all the properties continue to hold - except of
course the new fixed points at the starting points of the additional
paths. 

Finally, we conclude that it is \PPAD-hard to find an approximate
fixed point of $f$.
\end{proof}

\subsection{From Brouwer to WSNE}

Our third and final step in the proof reduces the problem of finding
approximate fixed points in $f$ to that of finding approximate well-supported
Nash equilibria. The reduction we use is based on \cite{Shm12_blog}
and appears almost exactly in this format in \cite{Bab13_query_complexity}.
\begin{prop}
(Essentially \cite{Bab13_query_complexity}) {\sc SuccinctBrouwer}$\left(n,M,\epsilon\right)$
$\leq_{P}$ {\sc SuccinctNash}$\left(2n,k+1,\frac{3}{4k^{2}}\right)$,
where $k=\lceil\frac{3+M}{\epsilon}\rceil$\end{prop}
\begin{proof}
We construct a game with two groups of $2n+2$ players each. The action
set of each player corresponds to $\left\{ 0,1/k,\dots,1\right\} $.
We denote the choice of strategies for the first group $\mathbf{a}=\left(a_{1}\dots a_{2n+2}\right)$,
and $\mathbf{b}=\left(b_{1},\dots b_{2n+2}\right)$ for the second
group. 

Each player in the first group attempts to imitate the behaviour of
the corresponding player in the second group. Her utility is given
by 
\[
u_{i}\left(a_{i},b_{i}\right)=-\left|a_{i}-b_{i}\right|^{2}
\]
The second group players attempt to imitate the value of $f$, when
applied to the vector of actions taken by all the players in the first
group. The utility of the $i$-th player is
\[
v_{i}\left(b_{i},\mathbf{a}\right)=-\left|f_{i}\left(\mathbf{a}\right)-b_{i}\right|^{2}
\]

Observe that when the $i$-th player on the second group (henceforth,
player $\left(i,2\right)$) applies a mixed strategy, the expected
utility for player $\left(i,1\right)$ (the $i$-th player in the
first group) is given by: 
\[
\mathbb{E}\left[u_{i}\left(a_{i},b_{i}\right)\right]=-\left|a_{i}-\mathbb{E}\left(b_{i}\right)\right|^{2}-\mathbf{Var}\left(b_{i}\right)
\]

Let $\alpha_{i}\in\left\{ 0,1/k,\dots,1\right\} $ be such that $\mathbb{E}\left(b_{i}\right)\in\left[\alpha_{i},\alpha_{i}+1/k\right]$,
and assume wlog that $\mathbb{E}\left(b_{i}\right)\in\left[\alpha_{i},\alpha_{i}+1/2k\right]$.
Then we can lower-bound the expected utility of player $\left(i,1\right)$
when playing $\alpha_{i}$: 
\[
\mathbb{E}\left[u_{i}\left(\alpha_{i},b_{i}\right)\right]\geq-\frac{1}{4k^{2}}-\mathbf{Var}\left(b_{i}\right)
\]
On the other hand, for any $\gamma\notin\left\{ \alpha_{i},\alpha_{i}+1/k\right\} $,
\[
\mathbb{E}\left[u_{i}\left(\gamma,b_{i}\right)\right]\leq-\frac{1}{k^{2}}-\mathbf{Var}\left(b_{i}\right)
\]

Therefore in every $\frac{3}{4k^{2}}$-NE, the support of the mixed
strategy of player $\left(i,1\right)$ is restricted to $\left\{ \alpha_{i},\alpha_{i}+1/k\right\} $.

For the second group of players, we have 
\[
\mathbb{E}\left[v_{i}\left(b_{i},\mathbf{a}\right)\right]=-\left|\mathbb{E}\left(f_{i}\left(\mathbf{a}\right)\right)-b_{i}\right|^{2}-\mathbf{Var}\left(f_{i}\left(\mathbf{a}\right)\right)
\]
Let $\beta_{i}\in\left\{ 0,1/k,\dots,1\right\} $ be such that $\mathbb{E}\left(f_{i}\left(\mathbf{a}\right)\right)\in\left[\beta_{i},\beta_{i}+1/k\right]$,
then in every $\frac{3}{4k^{2}}$-NE, the support player $\left(i,2\right)$
is restricted to $\left\{ \beta_{i},\beta_{i}+1/k\right\} $.

Finally, given a $\frac{3}{4k^{2}}$-NE, we use the Lipschitz property
of $f$ to derive an approximate fixed point. Notice that for $\mathbf{\alpha}$
and $\mathbf{\beta}$ as defined above, 
\[
\left|\alpha_{i}-\beta_{i}\right|\leq\left|\alpha_{i}-\mathbb{E}\left(b_{i}\right)\right|+\left|\mathbb{E}\left(b_{i}\right)-\beta_{i}\right|\leq\frac{2}{k}
\]

Likewise,
\[
\left|\beta_{i}-f_{i}\left(\mathbf{\alpha}\right)\right|\leq\left|\beta_{i}-\mathbb{E}\left(f_{i}\left(\mathbf{a}\right)\right)\right|+\left|\mathbb{E}\left(f_{i}\left(\mathbf{a}\right)\right)-f_{i}\left(\mathbf{\alpha}\right)\right|\leq\frac{1}{k}+\frac{M}{k}\mbox{ ,}
\]
where $M$ is the Lipschitz constant of $f$.

Therefore $\left|\mathbf{\alpha}-f\left(\mathbf{\alpha}\right)\right|_{\infty}\leq\frac{3+M}{k}$,
so $\alpha$ is a $\frac{3+M}{k}$-approximate fixed point of $f$.
\end{proof}

\section{Open problems}

We prove that finding an $\epsilon$-well-supported Nash equilibrium
\PPAD-hard even for constant $\epsilon$. This is a modest step towards
a better understanding of the computational complexity of approximating
Nash equilibrium in games with a large number of players. Many important
questions remain open:

\paragraph{NE vs WSNE}

As mentioned earlier, in the domain of constant approximations, Nash
equilibria may be strictly easier than finding well-supported Nash
equilibria. 

{\em What is the complexity of finding $\epsilon$-ANE?}

We note that a similar obstacle was encountered by Babichenko \cite{Bab13_query_complexity},
in the case of query complexity, for essentially the same reasons.

\paragraph{Simpler games}

Other, more restrictive forms of succinct games have been studied
quite extensively before (see e.g. \cite{Pap07_succint-games_inbook}).
It would be interesting to improve our understanding of the complexity
of those games. In particular, 

{\em What is the complexity of finding an $\epsilon$-NE in bounded-degree
graphical games?}

\begin{comment}
This question seems to be tightly related to the previous question
about the complexity of $\epsilon$-ANE.
\end{comment}

\paragraph{Upper bound on complexity}

Embarrassingly, we do not have a (non-trivial) upper bound on the
complexity of finding an $\epsilon$-NE. One particular obstacle to
showing that this problem belongs to the class \PPAD, is that it
is not even clear how to find an approximate best response to mixed
strategies without using randomness. (In particular, by \cite{SV12_succint_games},
this problem is also \BPP-hard; we do not know whether \BPP\;is
contained in \PPAD.) 

{\em What is the right complexity classification of the problem of
finding an $\epsilon$-NE?}

\section{Addendum}

Since the first version of this paper was posted, we found solutions
to the first two open problems mentioned above \cite{Rub14b_simpler-games}.

\bibliographystyle{plain}
\bibliography{succint_games}

\end{document}